\documentclass[11pt]{article}

\usepackage{amssymb}

\newtheorem{thm}{Theorem}

\newtheorem{definition}[thm]{Definition}
\newtheorem{proposition}[thm]{Proposition}

\newenvironment{proof}{\noindent\bf{Proof.}\rm}{\hfill$\blacksquare$\bigskip}

\def\opt{{\mbox{opt}}}

\begin{document}

\title{Mechanism design with uncertain inputs
\\
(to err is human, to forgive divine)}

\author{Uriel Feige~\thanks{Department of Computer Science and Applied Mathematics, Weizmann Institute, Rehovot, Israel. {\tt uriel.feige@weizmann.ac.il.} Work done at Microsoft R\&D Center, Herzelia, Israel.}
and Moshe Tennenholtz~\thanks{Microsoft R\&D Center, Herzelia, Israel, and Faculty of Industrial Engineering and Management, Technion, Haifa, Israel.  {\tt moshet@ie.technion.ac.il}}}

\maketitle

\begin{abstract}
We consider a task of scheduling with a common deadline on a single machine. Every player reports to a scheduler the length of his job and the scheduler needs to finish as many jobs as possible by the deadline. For this simple problem, there is a truthful mechanism that achieves maximum welfare in dominant strategies. The new aspect of our work is that in our setting players are uncertain about their own job lengths, and hence are incapable of providing truthful reports (in the strict sense of the word). For a probabilistic model for uncertainty our main results are as follows.

\begin{enumerate}

\item Even with relatively little uncertainty, no mechanism can guarantee a constant fraction of the maximum welfare.

\item To remedy this situation, we introduce a new measure of economic efficiency, based on a notion of a {\em fair share} of a player, and design mechanisms that are $\Omega(1)$-fair. In addition to its intrinsic appeal, our notion of fairness implies good approximation of maximum welfare in several cases of interest.

\item In our mechanisms the machine is sometimes left idle even though there are jobs that want to use it. We show that this unfavorable aspect is unavoidable, unless one gives up other favorable aspects (e.g., give up $\Omega(1)$-fairness).

\end{enumerate}

We also consider a qualitative approach to uncertainty as an alternative to the probabilistic quantitative model. In the qualitative approach we break away from solution concepts such as dominant strategies (they are no longer well defined), and instead suggest an axiomatic approach, which amounts to listing desirable properties for mechanisms. We provide a mechanism that satisfies these properties.
\end{abstract}

\newpage

\section{Introduction}

The main theme of this paper is that of mechanism design in situations in which players are uncertain about their own private inputs. Rather than address this issue in full generality, we concentrate here on one particular task, in hope that the insights provided from its study will be relevant for other tasks as well. We call this task SSCD (Strategic Scheduling with a Common Deadline).
There are $n$ players $P_1, \ldots, P_n$. Every player has a single job, with job $J_i$ belonging to player $P_i$. The length of job $J_i$ is denoted by $\ell_i$, and we assume that the value of $\ell_i$ is given as private input to player $P_i$. There is a single machine that can process the jobs one at a time, and a common deadline $D$. The deadline can be interpreted as the machine being active for $D$ time steps, and not available afterwards. Each player wants his job to finish before the deadline. Hence we view the utility function of players to have value~1 if their own job finishes by time $D$, and~0 otherwise. There is a {\em scheduler} that schedules the jobs on the machine. The operations available to the scheduler are to start a job, abort a job, and in cases where preemption is allowed, to resume a previously aborted job. For deciding which operations to perform, the scheduler can solicit private information from the players (ask them to report the job lengths), and observe when a job in completed (and hence schedule a new job).
The goal of the scheduler is to maximize social welfare, which in our case means to have as many jobs as possible finish before time $D$.

SSCD can easily be solved in dominant strategies in the setting described above. Each player is asked to report the length of his job. Let the reported lengths be $r_1, \ldots, r_n$. Sort the jobs in order of increasing $r_i$, and run in this order each job $J_i$ for $r_i$ time steps, until $D$ time steps are reached. If all reports are truthful, meaning $r_i = \ell_i$, then this schedule runs the shortest jobs first, and will indeed complete the maximum possible number of jobs in the allotted time. Moreover, being truthful is a dominant strategy for all players. If  a player reports  a value $r_i$ smaller than $\ell_i$, he will not be given sufficient time to complete his job (and his payoff will necessarily be~0). If a player reports a value $r_i$ larger than $\ell_i$, this may only delay the time his job is scheduled, and hence cannot increase his payoff. (We remark that this mechanism is not immune to sibling attacks. A long job might be split into several smaller parts, and each part might be submitted to the scheduler on behalf of a different player. In our current simplified model we assume no sibling attacks.)

In the current work we investigate situations in which players have some uncertainty about their private inputs. For SSCD, this means that a player has some doubts whether his private value $\ell_i$ is the true length of his job. Hence the value of $\ell_i$ is a belief, rather than actual knowledge. Such doubts are reasonable in many contexts (if $J_i$ is a computer program, it may be very difficult to estimate beforehand how long it will run). It is instructive to discuss informally what might happen in this case with the shortest first mechanism. Suppose that the players believe that the distribution of lengths of jobs is such that the scheduler can easily schedule the shorter jobs, but will not have time to schedule the longer jobs. Then it is natural for owners of shorter jobs to report a value of $r_i$ that is somewhat larger than their believed $\ell_i$. They would not want a slight error in their beliefs to cause their jobs to abort after $\ell_i$ steps, when in fact the jobs are only slightly larger and could easily have finished in a few additional steps. On the other hand, owners of long jobs will be tempted to report a short length. They have nothing to lose (reporting a long length they would not be scheduled at all), and they may potentially gain (as they are not really certain about the length of their jobs, and maybe their jobs are short). Such strategic reports might lead to a situation in which the scheduler schedules the long jobs first, gives them insufficient time, and eventually no job finishes.

The main aspect that the example above illustrates is that mechanisms that perform perfectly when players have accurate information about their private inputs may perform arbitrarily bad when players have uncertainty about their private inputs. This of course is not a new observation, and can be taken as folklore. The new aspect of our work is addressing this issue in a formal way.

Let us elaborate on which kinds of private input uncertainty refers to in our context. The private input to a player can be of two kinds. One is subjective, with the utility function of a player serving as a prime example. In much of game theory literature, the {\em type} of a player refers only to this aspect. However, the kind of private input that we refer to is objective. It is a parameter (length of job) that is associated with a player that is initially private, but eventually may become public. Hence if a player misreports the value of this parameter, then it is in principle possible to eventually detect this, by comparing the report with the actual truth. This allows one to introduce ``punishments" into mechanisms, and in extreme cases these punishments might be disproportional to the harm done by misreporting (presumably, under the assumption that rational players will never be subject to them). These punishment mechanisms are very effective in theoretical mathematical studies, but might be inappropriate in real life situations in which uncertainty might be inherent. Recalling the saying of Alexander Pope, {\em to err is human, to forgive divine}, we acknowledge that errors in the reports of players will happen, and wish the mechanism to be {\em forgiving}: small errors on behalf of players should not lead to severe consequences.

Forgiveness with respect to misreporting of private inputs is not a goal by itself, but rather a useful principle towards the design of better mechanisms. It should be used in a clever way that discourages manipulations by the players, and leads to desirable performance guarantees.

\subsection{The underlying model}

There is more than one way of interpreting uncertainty about private inputs, and the different interpretations may lead to different mechanisms.

Job lengths can be modeled in one of three ways:

{\bf Deterministic length.} Each job has a fixed length. Given sufficient resources, the player would be able to determine the exact length of his job. The uncertainty in length is a result of complexity: it might be unrealistic to expect  the player to invest the resources needed in order to obtain the exact information. However, it may be realistic to expect the player to provide an estimate of the length of the job.

{\bf Random length.} The length of a job is actually a random variable. This may be the case when the job itself is a randomized algorithm: the algorithm itself is known, but the random coin tosses (which affect the running time) are performed only at run time. Hence there is some probability distribution that exactly characterizes the length of the job.

{\bf Nondeterministic length.} The length of the job may depend on various aspects that are not known beforehand (e.g., how fragmented will the disk space be at the time the job is run).

Our work is mainly relevant to situations that resemble deterministic or random job lengths, and not very relevant to the nondeterministic setting.

The private input of players (their beliefs) can be modeled in one of two ways:

{\bf Quantitative private input.} The private input is a probability distribution over job-lengths. For random length jobs, the private input may (or may not) actually coincide with the true probability distribution. For other models of job lengths, the private input models beliefs of the player, based on partial information available to him. We need not assume that the player possesses an explicit probability distribution. It suffices that the player only has preferences over various lotteries (e.g., prefers to get $t_1$ time units with probability $p_1$ over getting $t_2$ time units with probability $p_2$), and if these preferences are consistent (in some well defined way), this is mathematically equivalent to having a probability distribution over lengths. (This last line of argument is similar to that used by Von Neumann and Morgenstern in defining utility functions.)

{\bf Qualitative private input.} The private input is a single job length. It is interpreted as an estimate of the true length, but there is no quantitative model explaining to which extent the player trusts this estimate. In particular, the preference of a player over various lotteries might be inconsistent with any probability distribution over lengths.

Most of our work will be concerned with the combination {\em random length / quantitative private input}, and specifically the special case when the private input coincides with the true distribution. However, we shall also briefly address the case of {\em deterministic length / qualitative private input}.

In our setting, the players report lengths (or distributions) for their jobs, and the scheduler produces a schedule. We distinguish between three types of reports.

{\bf Honest report.} A player reports his private input.

{\bf Rational report.} A player reports whatever is perceived to maximize his payoff.

{\bf Malicious report.} The report of the player might be arbitrary.

In general, we seek mechanisms in which the honest report is rational. For quantitative private inputs, this will most often mean implementation in dominant strategies, rather than Nash equilibrium. Solution concepts based on equilibria require extensive knowledge of the players about their environment (what strategies are other players playing, what are the true lengths of jobs of other players), and hence might be questionable in situations where players are uncertain even about their own inputs. For qualitative private inputs, the notion of a rational strategy is not well defined, and we shall need to explore other solution concepts.

At any time, only one job is run on the machine. Upon receiving reports from the players, a {\bf nonadaptive} scheduler allocates time intervals to the jobs, with total length of time intervals not exceeding $D$. A job {\em starts} at the beginning of its time interval, and is {\em aborted} at the end of the time interval (if not finished by that time). An {\bf adaptive} scheduler allocates time intervals to jobs one at a time, and if the job finishes before the end of its interval, the remaining time steps may be used for other jobs.
For schedulers with {\bf preemption} (which necessarily are also adaptive), a job that was previously aborted may be {\em resumed}.
We assume that there is no overhead involved in preemption. We do not claim that this model of preemption is realistic, but we do find it informative to consider it. A scheduler may be {\bf deterministic}, in which case the job currently run is determined only based on the reports of players and actual observed running times; or {\bf randomized}, where the scheduler can in addition base its decisions on random coin tosses. Most of our mechanisms are randomized and involve no preemption.

We remark that our current study is concerned only with mechanisms that do not involve monetary transfer. A player cannot offer to pay for the right to run his jobs for a certain number of time units.

\subsection{Main results for random length / quantitative private input}

We use the following notation. There are $n$ players and a single machine that can run for a total of $D$ units of time. The private input to player $P_i$ is a probability distribution $f_i$ and is assumed to be correct: for each $j$, $f_i(j)$ specifies the true probability with which job $J_i$ finishes in at most $j$ time steps. The goal of the mechanism is to schedule the jobs so as to maximize the number of jobs that finish within the allowed $D$ time steps (or maximize the expectation of this number in probabilistic settings). Recall that for the case of deterministic jobs (and quantitative private input that is accurate), maximum welfare can be achieved in dominant strategies by the shortest first mechanism. Our first result concerns the case in which there is only little uncertainty about job lengths.

\begin{definition}
\label{def:neardeterministic}
A job is {\em near deterministic} if there are only two values that its length can have, and they differ from each other by a factor of at most~2.
\end{definition}

The choice of a factor of at most 2 in Definition~\ref{def:neardeterministic} is to a large extent arbitrary. Any constant factor larger than 1 will give results in the same spirit as those that we prove below.

\begin{thm}
\label{thm:neardeterministic}
For every strategic scheduler (even randomized, with preemption, and using Nash-equilibrium as a solution concept), there are collections of near deterministic jobs for which $2^{(O{\sqrt{\log n}})}$ jobs can finish, whereas the expected number of jobs that the scheduler completes is $O(1)$.
\end{thm}

The bounds in Theorem~\ref{thm:neardeterministic} are essentially tight -- we show mechanisms that guarantee an approximation ratio of $2^{(O{\sqrt{\log n}})}$ for the maximum welfare (in the near deterministic / random length / quantitative private input / dominant strategies setting).

Theorem~\ref{thm:neardeterministic} illustrates that even with small uncertainty in private inputs, no mechanism has good performance guarantees with respect to welfare. In a sense, this is a theorem about {\em price of strategic uncertainty} (PoSU): if either uncertainty or strategic behavior (as opposed to honest reports) is removed, this price need not be paid. This may be viewed as an extension of previous notions that refer only to strategic behavior, such as {\em price of anarchy} and {\em price of stability}.

Given the large PoSU with respect to welfare, we suggest a different performance measure. The guarantee is per player, rather than some aggregation in the form of total welfare. The guarantee corresponds to the probability of the job finishing if all other jobs were identical to it and the optimal schedule was used. We call this the {\em fair share} of the player. To simplify the presentation that follows, we shall use an approximate notion of fair share that we shall now introduce. Recall that the running time of a player is given in form of a probability distribution $f_i$. If all players have the same $f_i = f$, we may consider a nonadaptive scheduler that we call {\em canonical}: compute $t \ge D/n$ that maximizes the probability of success per time step, namely $f(t)/t$, and run $D/t$ jobs, each for $t$ time steps. Proposition~\ref{pro:fairshare} shows that the canonical scheduler guarantees in expectation a constant fraction of the optimal expected welfare, even compared to schedules that use preemption.

\begin{proposition}
\label{pro:fairshare}
Given $n$ identical jobs with probability distribution $f$ and deadline $D$, the welfare of the optimal schedule (even with preemption) is at most~3 times that of the canonical schedule.
\end{proposition}

Computing the {\em fair share} with respect to the canonical scheduler (rather than with respect to the optimal one) gives the expression appearing in Definition~\ref{def:fairshare}.

\begin{definition}
\label{def:fairshare}
Given $n$ jobs and deadline $D$, the {\em fair share} for player $P_i$ with private input $f_i$ is defined to be  $\max_t [(f_i(t)/t) \cdot (D/n)]$, where the maximization is over the range $D/n \le t \le D$.
\end{definition}

\begin{definition}
For $0 < \rho \le 1$, a scheduler is $\rho$-fair if for every player it guarantees a finishing probability of at least $\rho$ times the player's fair share.
\end{definition}

We view the notion of $\rho$-fair schedulers as an adequate replacement to the notion of schedulers that maximize welfare. Fairness provides guarantees to players even if all other players are malicious (or totally uninformed regarding the true lengths of their jobs). Schedules that attempt to maximize welfare need not provide any guarantees in these cases. Technically, maximizing fairness can be incompatible with maximizing welfare. In the former case long jobs must be given some probability of finishing, whereas in the latter they must be discarded. However, fairness does correspond to good approximation of maximum welfare whenever the sum of fair shares is close to the maximum welfare. This happens in many natural cases, e.g., when private inputs of players are similar to each other, or if the instance is such that the optimal welfare is linear in $n$. (The proofs of these simple observations are omitted.)

Our notion of fairness is different from some classical notions of fairness (such as proportional division~\cite{BramsTaylor}, or max-min payoff~\cite{fairness}) that postulate that different players are entitled to equal shares of the payoff. In our notion, players that do not place high demand on the resources of the machine are entitled to higher payoffs. The setting is also different: we deal with a homogenous good (time on a machine) and nonlinear utility functions, whereas much of the previous literature (say, on cake cutting) deals with nonhomogeneous goods and linear utility functions.  Our notion of fair-share is inspired by the notion of {\em recoverable value} which was introduced in~\cite{FIMN} so as to get good approximations in interesting special cases for a problem that is hard to approximate in the worst case.

Fair schedulers are randomized by definition, as every player must be guaranteed some probability of finishing. If all $f_i$ are identical, no schedule can be more than $O(1)$-fair. In fact, $O(1)$ upper bounds on fairness hold also in other settings (implicit in the proof of Proposition~\ref{prop:nooblivious}). On the positive hand we have:

\begin{thm}
\label{thm:strategicfair}
There is a nonadaptive scheduler in dominant strategies that is $1/2$-fair.
\end{thm}

{\em Oblivious} schedulers ignore the reports of players and hence trivially involve dominant strategies. Theorem~\ref{thm:strategicfair} can be contrasted with the following theorem (proved in the appendix).

\begin{thm}
\label{thm:oblivious}
Every oblivious scheduler (even with preemption) is at most $O(1/\log n)$-fair. There is a nonadaptive oblivious scheduler that is $\Omega(1/\log n)$-fair.
\end{thm}

Our next set of results concerns another property of interest for schedulers.

\begin{definition}
A scheduler is {\em complete} if it has the property that whenever the sum of true job lengths is at most $D$, all jobs finish. When pre-emption is allowed, completeness is equivalent to not having idle time in which the machine is free and yet no unfinished job is scheduled on it.
\end{definition}

It is desirable that schedulers are complete, as it is difficult to justify why a machine is kept idle when there is demand for it. To present our result for complete schedulers, we provide one more definition.

\begin{definition}
A strategic scheduler is {\em trivial} if for every player there is a report that maximizes its expected probability of finishing, regardless of its $f_i$.
\end{definition}

Trivial strategic schedulers can simply be replaced by {\em oblivious} schedulers that ignore the reports of players, by assuming that players always report their maximizing report.

\begin{thm}
\label{thm:nocomplete}
\begin{enumerate}
\item Any complete scheduler with two players and dominant strategies (even allowing preemption) is trivial.
\item For three or more players, there are complete schedulers (with dominant strategies and no preemption) that are nontrivial.
\item Complete schedulers (with no preemption) cannot guarantee $(10/n)$-fairness in dominant strategies.
\end{enumerate}
\end{thm}

With preemption there are complete schedulers that guarantee $\Omega(1/\log n)$-fairness in dominant strategies. (E.g., in a first phase run the oblivious scheduler of Proposition~\ref{thm:yesoblivious}. In a second phase run all remaining jobs.) If in addition one gives up dominant strategies, $\Omega(1)$-fairness can be achieved. The following proposition is proved in the appendix.

\begin{proposition}
\label{thm:complete}
For any number of players there is a scheduler with preemption that is complete, in which honest reporting is an ex-ante Nash equilibrium, and in this equilibrium the schedule is $\Omega(1)$-fair.
\end{proposition}

\subsection{Main results for deterministic length / qualitative private inputs}

Oblivious schedulers clearly maintain their performance guarantees when players have qualitative rather than quantitative private inputs (since the reports of players are ignored). To improve over oblivious schedulers, one needs to assume that the qualitative private inputs are actually good estimates to the length of the jobs, and moreover, to design a mechanism that makes it desirable for players to reveal their private inputs. This last requirement is a conceptual challenge, because the qualitative definition of private inputs does not suffice in order to offer predictions of how a player will react to incentives. We circumvent this conceptual difficulty by instead offering an axiomatic approach that lists a set of properties that it is desirable that mechanisms have. These properties are guarantees to a player that hold regardless of the reports and actual job lengths of other players.

\begin{enumerate}

\item Error-monotonicity. The player's expected payoff is maximized if his report is the true length, and decreases monotonically as the error grows.

\item Error-symmetry. If the true length of a job is denoted by $\ell_i$, then for every error $e > 0$, the expected payoff when reporting $\ell_i + e$ is the same as the expected payoff when reporting $\ell_i - e$.

\item Forgiveness. Except for the requirement to keep error-monotonicity and error-symmetry, ``punishment" in response to error is kept to a minimum.

\end{enumerate}

The above properties provide incentives to a player to report his best estimate for the length of his job, not to bias his estimate in either direction (as one might do in shortest first schedule), and to actually participate (trusting that even if his report is not accurate, reasonable payoff will remain). The question of whether given these incentives the player will actually provide an honest report is a psychological question, and not a mathematical question (since we do not assume a quantitative model), and hence beyond the scope of this paper.

As previously explained, maximizing welfare and achieving fairness can be incompatible. Hence we present two mechanisms for {\em deterministic length / qualitative private inputs} that offer different types of guarantees.

\begin{thm}
\label{thm:qualitative}
There is a randomized mechanism (without preemption) satisfying error-monotonicity, error-symmetry and forgiveness, that in addition is $\Omega(1)$-fair if all reports happen to be correct.
\end{thm}

\begin{thm}
\label{thm:preemption}
There is a deterministic mechanism with preemption satisfying error-monotonicity, error-symmetry and forgiveness, that in addition achieves maximum welfare if all reports happen to be correct.
\end{thm}

\subsection{Related work}

Our work is broadly related to work in algorithmic mechanism design \cite{NR2001}, and in particular to a more recent line of research on approximate mechanism design without money \cite{PT09}. It tries to quantify loss of efficiency in strategic settings, in the spirit of work on the price of anarchy \cite{KP99}, though the source of loss in efficiency in our case is a combination of two aspects -- selfishness and uncertainty about private inputs. We are aware of some previous work \cite{PRST08} dealing with uncertain private inputs. However our work differs considerably from this previous work in that our mechanisms do not involve money, and we introduce new performance measures ({\em fair-share}) that are less sensitive to uncertainties. In our work we use linguistic terms such as {\em fairness} and {\em forgiveness}. Similar terms have been used in other contexts (e.g., in~\cite{BramsTaylor}), but often not with the same interpretation given in our work. More detailed discussion of related work is left to the appendix.

\section{Proofs -- Random length / quantitative private input}

Before proving Theorem~\ref{thm:neardeterministic}, it is instructive to consider some simpler settings. We start with the oblivious setting, in which players are not asked to report their private inputs.

The trivial oblivious schedule is as follows. Sort jobs in a random order. Run them one by one in this order, switching to a new job whenever the previous job finishes.

The performance guarantee of this schedule is rather poor. Consider for example the case in which half the jobs require time $2D/n$ each, whereas the other half require time $D$ each. The optimal schedule can schedule $n/2$ jobs, whereas the trivial oblivious finishes less than two jobs in expectation.

To remedy this situation, it is natural to introduce timeouts. If a job does not finish within the timeout period, it is aborted and the schedule moves to the next job. A timeout value of $D/\sqrt{n}$ ensures that at least $\sqrt{n}$ jobs have an opportunity to run, or equivalently, that each job has probability at least $1/\sqrt{n}$ of being run. In any schedule, at most $\sqrt{n}$ jobs run for more that $D/\sqrt{n}$ steps. This implies that the expected number of jobs completed by the timeout schedule is at least $(\opt - \sqrt{n})/\sqrt{n} = \opt/\sqrt{n} - 1$.

This form of a guarantee is essentially best possible for oblivious algorithms. Consider the case that $\sqrt{n}/2$ jobs each take time $2D/\sqrt{n}$, and all other jobs take time $D$. Any schedule can run at most $\sqrt{n}/2$ jobs for time $2D/\sqrt{n}$ each, and if the schedule is oblivious, the expected number of these jobs that finish is $\sqrt{n}/2 \cdot \sqrt{n}/2n = 1/4$, which can be viewed as $\opt/\sqrt{n} - 1/4$.

The above line of reasoning can be extended to the strategic setting. The proof of the following proposition illustrates one of the principles that will later be used in the proof of Theorem~\ref{thm:neardeterministic}.

\begin{proposition}
\label{prop:nooblivious}
In the strategic setting (with random length and quantitative private input), no scheduler can guarantee a welfare of more than $O(\opt/\sqrt{n} + 1)$.
\end{proposition}

\begin{proof}
Fix an arbitrary (randomized) strategic scheduling mechanism and assume for simplicity of notation that $D = n$. The revelation principle implies that under standard solution concepts (e.g., dominant strategies, Nash equilibrium), we may assume that players simply report their true private inputs (and the mechanism itself performs on behalf of the players whatever strategic decisions are best for them). Consider two scenarios. In one there are $n/t$ short jobs whose $f_i$ is distributed uniformly between 1 and $2t$, and $n - n/t$ long jobs whose $f_i$ is distributed uniformly between 1 and $n$. In the other there are $n/t + 1$ short jobs as above and $n - n/t - 1$ long jobs. Let $E_s$ be the expected amount of time that scheduler gives to a short job in the second scenario, and $E_l$ be the expected amount of time given to a long job in the first scenario. If $E_s > E_l$, then if all jobs report truthfully in the first scenario, there is incentive for long jobs to report instead that they are short. Hence for truthful reporting to be a Nash equilibrium, we must have that $E_s < E_l$. But by symmetry, $E_l \le n/(n - n/t - 1) < 2$ implying that $E_s < 2$. Hence in the first scenario, the expected number of short jobs that finish is at most $2/2t \cdot n/t$. When $t = \sqrt{n}$, the guarantee of a strategic scheduler is no better than $O(\opt/\sqrt{n} + 1)$.
\end{proof}

We now prove Theorem~\ref{thm:neardeterministic}, which unlike Proposition~\ref{prop:nooblivious} requires that jobs are near-deterministic.

\begin{proof} {\bf [Of Theorem~\ref{thm:neardeterministic}.]}
Let $D$, the total time available, be a power of~2, and let $d = 1 + \log D$. Let $k$ be a parameter that we choose as $k = 2D$, and let $\epsilon$ be small (we set $\epsilon = 2/k = 1/D$). Consider a collection $C$ of jobs that are arranged in groups $J_i$, for $0 \le i \le d$. Group $J_0$ has $D$ jobs, each of length~1. For $i \ge 1$, group $J_i$ has $Dk^i$ jobs with a small uncertainty about their length: each job has length $2^{i-1}$ with probability $\epsilon$ and length $2^i$ with probability $1 - \epsilon$. The optimal schedule runs all jobs of group $J_0$ for a payoff of $D$. 

Consider now an arbitrary (possibly randomized) scheduler for the strategic setting. With respect to this scheduler, and assuming that all players are truthful, we consider the following random variables. For $0 \le i \le d$, let $\ell_i$ be the expected number of jobs of group $J_i$ that get running time $2^{i-1}$ and let $u_i$ be the expected number of jobs of group $J_i$ that get running time $2^i$. For $J_0$ we fix $\ell_0 = 0$, and clearly, for $J_d$ we have $u_d = 0$ (because $2^d > D$).

Observe that all $\ell_i$ combined contribute only one to the expected number of jobs that finish, as $\sum \ell_i \le D$ and the probability that such a job finishes is $\epsilon = 1/D$.

Suppose that for every $i$ we have that $u_i \le 2^{-i}$. In this case all $u_i$ combined contribute only two to the expected number of jobs that finish, as $\sum u_i \le \sum_{i=0}^{d} 2^{-i} < 2$.

Hence to complete in expectation more than three jobs it is necessary that $u_i > 2^{-i}$ for some $i < d$ (recall that $u_d = 0$). We show that no such $i$ exists if the players have dominant strategies.

Assume for the sake of contradiction that $u_i > 2^{-i}$, and consider a collection $C'$ of jobs that differs from the original collection in the sense that there is one fewer job in $J_i$ and one more job in $J_{i+1}$. We claim that if all players are truthful, then the respective $\ell'$ and $u'$ random variables for $C'$ are related to the original $u_i$ via the inequality $\epsilon \ell'_{i+1} + u'_{i+1} \ge k\epsilon u_i$. The left hand side (divided by $1 + Dk^{i+1}$) represents the probability for a job in $J'_{i+1}$ that it finishes in $C'$, whereas the right hand side (divided by $Dk^{i+1}$) represents the probability for a job in $J'_{i+1}$ to finishes if it impersonates as a job of $J_i$ (and hence the scheduler behaves as if the collection is $C$). Any dominant strategy for the player in $C'$ must give it at least the same probability of finishing as the strategy of impersonating a player from a different group.

Consider the implications of the inequality $\epsilon \ell'_{i+1} + u'_{i+1} \ge k\epsilon u_i$. Observe that $\ell'_{i+1} \le D2^{-i}$, since the total time is $D$. Using $k = 2D$ it follows that $\epsilon \ell'_{i+1} \le k\epsilon 2^{-(i+1)} \le k\epsilon u_i/2$. It follows that $u'_{i+1} > k\epsilon u_i/2 = u_i = 2^{-i} > 2^{-(i+1)}$. Now we can consider yet another collection of jobs that differs from $C'$ in that it has one less job in $J_{i+1}$ and one more job in $J_{i+2}$. Iterating and repeating the argument as above we eventually get to a collection of jobs for which $u_d > 0$, which is a contradiction.

By choosing $D \simeq 2^{\sqrt{\log n} - 1}$ we have $d = \sqrt{\log n}$, and the total number of jobs is roughly $n$ (ignoring low order multiplicative terms). In this case every strategic mechanism is restricted to expected payoff of~3 whereas the optimum is  $2^{\sqrt{\log n} - 1}$.
\end{proof}

The bounds in Theorem~\ref{thm:neardeterministic} are essentially tight.

\begin{thm}
\label{thm:goodneardeterministic}
There is a strategic scheduler with dominant strategies that for every collection of near deterministic jobs approximates the maximum welfare within a ratio of $2^{(O{\sqrt{\log n}})}$.
\end{thm}

The proof of Theorem~\ref{thm:goodneardeterministic} appears in the appendix. Recall that we have seen that oblivious schedulers cannot achieve guarantees as in Theorem~\ref{thm:goodneardeterministic} (not even for deterministic jobs). Hence the strategic scheduler in the proof of Theorem~\ref{thm:goodneardeterministic} cannot be oblivious.

We now prove Proposition~\ref{pro:fairshare}.

\begin{proof}
Observe that a probability function $f(t)$ is defined only for $0 \le t \le D$, is non decreasing and satisfies $0 \le f(t) \le 1$. We may assume that $f(0) < 1$  and $f(D) > 0$, as otherwise the proposition is trivial. Given $f$, let $g$ be the unique function that minimizes $g(D/n)$ under the following constraints:

\begin{enumerate}

 \item $g$ is a probability function as defined above.

 \item $g(t) \ge f(t)$ for all $0 \le t \le D$.

 \item $g(t) > 0$ for all $t > 0$.

 \item $g$ is continuous and piecewise linear with at most two pieces, the first with positive slope the second with slope~0. If both pieces exist, we call their meeting point $t_m$.

\item $g(t) = f(t)$ for some $D/n \le t \le D$ in the first piece. We call the least value of $t$ at which this holds $p$.

\end{enumerate}

Constraint~(2) implies that the welfare of the optimal schedule with respect to $g$ is at least as large as that for $f$.

We now describe the optimal schedule (with preemption) for $g$. Run the jobs in two phases. In the first phase, every job is run for an infinitesimal amount of time. In the second phase, those jobs that have not finished are run in an arbitrary order until time runs out. The first phase produces an expected welfare of $ng(0)$. Assuming without loss of generality that $g(D) = 1$ (if $g(D) < 1$ the bounds in our analysis only improve), the maximum number of time steps that a job is run in the second phase (if run at all) is precisely $\frac{1 - g(0)}{g(D/n) - g(0)} \frac{D}{n}$, and the expected number of time steps is thus $\frac{1}{2}\frac{1 - g(0)}{g(D/n) - g(0)} \frac{D}{n}$ (by linearity of $g$ in this range). If this last expression is smaller than $D/n$, then it is not hard to see that the canonical scheduler recovers at least half the welfare by setting $t = D/n$. Hence we shall assume that it is at least $D/n$.
Then the expected number of jobs finishing in the second phase is at most $2n\frac{g(D/n) - g(0)}{1 - g(0)}$ (at most because it is bounded by $n - ng(0)$), and in both phases combined it is at most $ng(0) + 2n\frac{g(D/n) - g(0)}{1 - g(0)}$.

For the canonical scheduler for $f$, let us consider two possible values for $t \ge D/n$. One is $t = D/n$, and then the canonical scheduler will recover expected welfare at least $nf(D/n) \ge d(g(0))$ The other is $t = p$ (from constraint~(5)) and then the condition $f(p) = g(p)$ can be seen to imply that the expected welfare is at least $n\frac{g(D/n) - g(0)}{1 - g(0)}$. At least in one of these cases it is at least one third that of the optimal schedule for $g$. As the choice of $t$ in the canonical scheduler is at least as good, the proposition is proved.
\end{proof}

The gap of~3 between optimal and canonical scheduler in Proposition~\ref{pro:fairshare} is essentially tight. For some large $k$ consider a probability function $f(t) = 1/k$ in the range $0 \le t \le D/n$, and $f(t) = d/Dk$ in the range $D/n \le t \le kD/n$. The canonical scheduler is indifferent to a choice of $t$ in the range $D/n \le t \le kD/n$ and will have expected welfare $n/k$. A preemptive scheduler with two phases as in the proof of Proposition~\ref{pro:fairshare} achieves expected welfare of roughly $1/k + 2/(k+1)$. Let us remark that the canonical scheduler is {\em not} the optimal nonadaptive scheduler (a better nonadaptive schedule is to choose two thresholds $0 \le t_1 \le D/n \le t_2$ and give some jobs $t_1$ time steps and other jobs $t_2$ time steps), and it is not the goal of Proposition~\ref{pro:fairshare} to characterize the gap between nonadaptive schedules and preemptive ones.

We now turn to prove Theorem~\ref{thm:strategicfair} (an $1/2$-fair scheduler).

\begin{proof} {\bf [Of Theorem~\ref{thm:strategicfair}.]}
Our proof uses a principle that may be applicable also in other situations. The {\em fair-share scheduler} shares equally the $D$ time slots among the $n$ players, but only in a sense of expectations. Each player may choose an arbitrary probability distribution over running times whose expectation does not exceed his share. The basic idea is as follows. Given a private input $f_i$, the player chooses (or alternatively, the player reports $f_i$ to the mechanism and the mechanism chooses) a value $t_i$ that maximizes the product of $f_i(t_i)$ (which is his probability of finishing in $t_i$ steps) and $p_i = \min[1,(D/n)/ t_i]$ (which will be the probability that the mechanism actually allows the player to run). Thereafter the mechanism chooses a random set of players to run, such that the probability of each player to be chosen is his respective $p_i$, the amount of time given to run if chosen is $t_i$, and the total running time committed to is exactly $D$. It is not hard to see that a mechanism with the above properties is truthful and 1-fair. Unfortunately, no such mechanism actually exists, due to divisibility issues. For example, if all $t_i$ equal $D/2 + 1$ it is impossible for the scheduler to satisfy the constraint that the total time committed to is $D$, and no scheduler can be more that $(1/2 + O(1/D))$-fair. However, the above ideas can be modified so as to get a mechanism that is $1/2$-fair. Details appear in the appendix.
\end{proof}

We now prove Theorem~\ref{thm:nocomplete} (limitations of complete schedulers).

\begin{proof} {\bf [Item (1) of Theorem~\ref{thm:nocomplete}].}
In the proof it is assumed that running times $t$ are positive integers. There are two players, $P_1$ and $P_2$, and a deadline $D$. Fix an arbitrary (randomized) complete scheduler $S$. Consider two arbitrary reports, $r_1$ for $P_1$ and $r_2$ for $P_2$. For every $t \le D$, define $p(t|r_1,r_2)$ to be the probability that $S$ gives $P_1$ at least $t$ time steps, if the reports are $r_1$ and $r_2$ and the true lengths of the jobs are $\ell_1 = \ell_2 = D$. Given two reports $r_1$ and $r_1'$, we say that $r_1$ {\em wins} against $r_1'$ with respect to $r_2$ if for every $1 \le t \le D$ we have $p(t|r_1,r_2) \ge p(t|r_1',r_2)$.
We say that $r_1$ {\em suffix-wins} against $r_1'$ with respect to $r_2$ if there is some $t\le D$ for which $p(t|r_1,r_2) > p(t|r_1',r_2)$,  and $p(t'|r_1,r_2) = p(t'|r_1',r_2)$ for every $t' > t$.

Given $r_2$, suffix-winning induces an order relation on reports $r_1$ of $P_1$. This order has a maximal element (because probability functions over finitely many outcomes form a compact set). This maximal element is unique up to equivalence (all maximal elements $r_1$ have the same probability function $p(t|r_1,r_2)$).

\begin{proposition}
\label{pro:suffixwin}
Let $S$ be a complete scheduler with dominant strategies, and let the private input to $P_1$ be a probability distribution with full support (his job may have any running time with nonzero probability). Then a dominant report for player $P_1$ (if there is one) must be at the top of the suffix winning order (with respect to any $r_2$).
\end{proposition}

\begin{proof}
Let $r_1'$ be a dominant report for player $P_1$, and assume for the sake of contradiction that there is some $r_1$ and $r_2$ such that $r_1$ suffix-wins against $r_1'$ with respect to $r_2$. Hence there is some $t\le D$ for which $p(t|r_1,r_2) > p(t|r_1',r_2)$,  and $p(t'|r_1,r_2) = p(t'|r_1',r_2)$ for every $t' > t$. Suppose that $P_2$ reports $r_2$ and the true length of his job is $\ell_2 = D - t + 1$. In this case $P_1$ get a running time of at least $t-1$ for certainty, and a larger running time with probability $p(t|r_1',r_2)$. But getting a larger running time with probability $p(t|r_1,r_2)$ would be strictly better, since it entails some gain if $\ell_1 = t$ and no loss otherwise. Hence $r_1'$ cannot be a dominant report.
\end{proof}

\begin{proposition}
\label{pro:win}
If $S$ is a complete scheduler with dominant strategies. Then for every $r_2$, there is some report $r_1$ that wins against all other reports $r_1'$.
\end{proposition}

\begin{proof}
Fix $r_2$. Let $r_1$ be a report at the top of the suffix winning order, and assume for the sake of contradiction that there is a report $r_1'$ such that $p(t|r_1,r_2) < p(t|r_1',r_2)$ for some $t$. Let $\epsilon = (p(t|r_1',r_2) - p(t|r_1,r_2))/2$. Consider a private input for $P_1$ with full support, giving probability $1 - \epsilon$ of having length $t$. Being full support, Proposition~\ref{pro:suffixwin} implies that $r_1$ is a dominant report. However, if $\ell_2 = D - t + 1$ then reporting $r_1'$ is preferable to $r_1$.
\end{proof}

Observe that in Proposition~\ref{pro:win} we can switch the order of quantifiers. That is, there is a report $r_1$ that wins against all $r_1'$ with respect to all $r_2$. Otherwise, for private inputs with full support, there would not be any dominant report for $P_1$, since the optimal report would depend on $r_2$. Hence $P_1$ can safely report this winning $r_1$  whenever his private input has full support. To complete the proof of Theorem~\ref{thm:nocomplete}, it remains to show that reporting $r_1$ is optimal even if the private input to $P_1$ does not have full support. But this is not difficult to see (e.g., by taking the private input to be a limit of private inputs with full support).
\end{proof}

\begin{proof} {\bf [Items (2) and (3) of Theorem~\ref{thm:nocomplete}].}
Sort the players in a random (or arbitrary) order $P_1, \ldots P_n$. The complete algorithm proceeds in rounds.
Let $n_i = n-i+1$ be the number of players that remain prior to round $i$, and let $D_i$ be the number of time steps that remain. Initially $n_1 = n$ and $D_1 = D$.
We now describe round $i$.

Preliminary operation: If $n_i \ge D_i$, let $P_i$ run until $n_i - 1$ time steps remain. (If it finishes, move to the next round.)
After the preliminary operation we have $D_i \le n_i - 1$. This means that even if $P_i$ is aborted, there still are sufficiently many other players to use up all remaining time steps.

Now offer $P_i$ to select a lottery with expectation~1.  Namely, $P_i$ chooses a value $t \le D_i$, and is offered $t$ time steps with probability $1/t$. Perform this lottery. If $P_i$ wins, let it run for $t$ time steps (or until he finishes, whatever happens first). Thereafter (and also if he loses), abort $P_i$ and move to the next round.

Clearly, this algorithm is nontrivial and complete. Note that it requires $n > 2$. If $n=2$, then in the preliminary operation player $P_1$ is run until only one time step is left. But then offering him a lottery with expectation~1 is giving him also the last time step. So the algorithm becomes equivalent to the oblivious algorithm of running the two players in a random order.

To prove item (3) of Theorem~\ref{thm:nocomplete}, consider the case when $D=n^2$ and there are $n$ identical jobs, each with probability $1/2$ has length $n$ and with probability $1/2$ has length uniformly distributed between~1 and $n^2$. The fair-share of a player is roughly~1/2. However, for every complete scheduler, once it starts running a job, it cannot abort until the number of time steps that remain is equal to the number of remaining players (as all remaining jobs might have length~1 and completeness will be lost). It is not difficult to see that the expected number of jobs completed is $O(1)$, implying that the scheduler is only $O(1/n)$-fair.
\end{proof}

\section{Proofs -- deterministic length / qualitative private input}

Consider scheduling mechanisms such as the shortest first scheduler or the fair share scheduler. Mistakenly reporting a value $r_i$ that is larger than the true $\ell_i$ has the following effect: either the scheduler does not run the job at all and hence does not detect that the report was erroneous, or the scheduler does run the job for $r_i$ steps, the job finishes earlier, and no punishment is involved. On the other hand, reporting a value $r_i$ that is too low typically has severe effects: failing to finish in $r_i$ steps, the job is aborted. Our qualitative mechanisms are all based on the principle that the loss in payoff to a player due to reporting a value of $r_i$ that is too low should be equalized with the loss in payoff to a player due to reporting a value of $r_i$ that is too high. We illustrate this idea in the case of the fair share scheduler, and this serves as a sketch of proof for Theorem~\ref{thm:qualitative}.

\begin{proof} {\bf [Of Theorem~\ref{thm:qualitative}].}
Recall the basic idea for the fair share scheduler. Each player reports his private input $t_i$. The scheduler considers the players in a random order. Let $u$ denote the ratio between the number of time steps available and number of players remaining upon reaching a player $P_i$. The scheduler tosses a biased coin and with probability $u/t_i$ lets the player run his job for $t_i$ steps. If the job finishes we are done. But what should the scheduler do if the job does not finish? In this case the scheduler tosses another biased coin, with probability $2/(t_i+2)$ aborts and with probability $t_i/(t_i+2)$ lets the job continue for one more step. What is the logic behind this? Assume that the true length of the job in indeed $t_i+1$. Then the probability of finishing is then $(u/t_i) \cdot t_i/(t_i+2) = u/(t_i+2)$. This would have been the probability of finishing had the player reported $(t_i+1) + 1$ , and now it is also the probability of finishing for the report $(t_i+1) - 1$. Such a mechanism has the desired properties of error monotonicity, error symmetry and forgiveness. By the same logic, if the job still does not finish at time $t_i+1$, it is given another extra step with probability $(t_i+2)/(t_i+4)$ and so on. Further details of how these ideas are combined with the fair-share scheduler to prove Theorem~\ref{thm:qualitative} are quite easy and omitted.
\end{proof}

Using ideas of a similar nature in combination with the shortest first schedule proves Theorem~\ref{thm:preemption}. Further details appear in the appendix.

\begin{appendix}

\section{Related work}

In this section we survey some work that may be related to our work. We shall mention work on scheduling, on price of anarchy, on mechanisms without money, on uncertain private inputs, on fairness, and on forgiveness. Our survey is not meant to be exhaustive.

{\bf Scheduling.}
Following the paper by Nisan and Ronen on algorithmic mechanism design \cite{NR2001} researchers have considered
the problem of machine scheduling in adversarial settings. Many settings have been considered.
For example, in \cite{LaviS07} there are $n$ jobs or tasks that need to be assigned to
$m$ machines, where each job has to be assigned to exactly
one machine. Assigning a job $j$ to a machine $i$ incurs a cost $c_{i,j}$ on machine $i$, and the load of a machine is
the sum of the loads incurred due to the jobs assigned to it. The goal is to schedule the jobs so as to minimize the maximum
load of a machine, which is termed the makespan of the
schedule. Each machine is assumed to be a strategic player who privately knows its own
processing time for each job, and may misrepresent these
values in order to decrease its load. Such problems are addressed via mechanism design;
the social designer, who holds the set of jobs to be
assigned, specifies, in addition to the schedule,
payments to the players in order to incentivize them to reveal
their true processing times. In all that rich literature, unlike in our work, monetary transfers are the central tool for design.

A relatively small part of the literature on scheduling in game theoretic settings deals
with deadlines. Porter \cite{Porter04} studies a problem of online scheduling of jobs on
a single processor. Each job is characterized by a release
time, a deadline, a processing time, and a value for successful
completion by its deadline. Monetary transfers are used in order to lead agents to a relatively efficient
behavior. On the other hand the work by Cres and Moulin \cite{CresMoulin} deals with a scheduling domain, where deadlines (rather than job lengths) are private information, and compares two scheduling policies that do not involve money. As in all other work on scheduling that we are aware of, a player does not have uncertainty about his private parameters.

{\bf Price of anarchy.}
The study of scheduling in the context of games with complete information has been advocated by the
 seminal work on the price of anarchy by Koutsoupias and
Papadimitriou \cite{KP99}. Many variants have been studied.
For example, in the classical problem of unrelated machine
scheduling we have $m$ parallel machines and $n$ independent jobs. Job $i$ induces a positive processing time (or load) $w_{i,j}$
 when processed by machine j. The load of a machine is the
total load of the jobs assigned to it. The quality of
an assignment of jobs to machines is measured by the
makespan (i.e., the maximum) of the machine loads or,
alternatively, the maximum completion time among all
jobs. The approach followed is both algorithmic and game-theoretic. Each job
is owned by a selfish agent. This gives rise to a selfish
scheduling setting where each agent aims to minimize
the completion time of her job with no regard to the
globally optimal schedule.
Such a selfish behaviour can lead to inefficient schedules from which no agent has
an incentive to unilaterally deviate in order to improve
the completion time of her job. One natural objective is to design
coordination mechanisms that guarantee that the assignments reached by the selfish agents are efficient (see e.g. [8],[23]).
The rich literature tackling the above issues, deal with complete information games, and agents' incentives to report private information does
not play a role

{\bf The use of money in mechanism design.}
Most work in algorithmic mechanism design deals with game-theoretic versions of optimization problems.
It has been observed \cite{CKV09} that there are two major classes of problems in algorithmic mechanism
design. The first class contains problems for which there exist optimal truthful mechanisms, but the
problem is computationally intractable. Typical examples include the line of work on combinatorial
auctions (see, e.g.,\cite{DNS06,RNDM,LaviS05,LOS02}), where the objective function is usually the maximization of the
social welfare, that is, the sum of agents' utilities. For this objective function a truthful optimal
mechanism is given by the Vickrey-Clarke-Groves (VCG) mechanism \cite{Vickrey,Clarke,Groves}.
VCG uses payments in order to align the interests of individual agents with the interests of society.
Unfortunately, it turns out that an approximation of the social welfare is insufficient to guarantee
truthfulness using VCG. Therefore, researchers have focused on designing truthful yet efficient
approximation mechanisms; in other words, researchers circumvent the computational hardness by resorting
to approximation, and at the same time enforce tailor-made payments to guarantee truthfulness.
Papers about scheduling on related machines (see, e.g., \cite{AAS07,AT01,DDDR08}) also fall into the first class,
although in the scheduling domain the objective is usually to minimize the makespan.
The second (significantly smaller) class of problems involves optimization problems which are
not necessarily intractable, but for which there is no optimal truthful mechanism. The prominent
problem in this class is scheduling on unrelated machines (see, e.g., \cite{CKV09,LaviS07}). In such domains one
might investigate the optimal approximation ratio achievable by any truthful mechanism, regardless
of computational feasibility.
Until recently, the assumption underlying essentially all previous work on truthful approximation mechanisms
was the existence of money, or, in other words, the ability to make payments.
Recently, Procaccia and Tennenholtz \cite{PT09} advocated approximated mechanism design without money. In that setting one considers
game-theoretic optimization problems where returning the optimal solution is not strategyproof.
The idea is that approximation can be
used to obtain strategyproofness without resorting to payments; In other words, achieving
strategyproofness, without using money, by compromising on  the optimality of the solution.
The approach has been applied in a variety of domains such as facility location \cite{AFPT10} and matching markets \cite{AFKP10}.

{\bf Uncertainty about private inputs:}
Most work on algorithmic mechanism design assumes that agents' valuations are private. However, there is also rich literature with regard to the
case that the agent may learn about his own valuation/cost by observing other agents' valuation/cost. The case of common values is a special case; in this case all agents share the same cost, but each of them may get different signals about that cost. The interested reader may learn about the related distinctions in e.g. \cite{Milgromprimer}. More related to our current work is the work on fault tolerant mechanism design
\cite{PRST08}. There the authors study the problem of
task allocation in which the private information of the agents is not only their costs to
attempt the tasks, but also their probabilities of failure. These probabilities reflect uncertainty on behalf of the agent, and are assumed to be independent of his own cost/valuation, and of other agents. For several different instances
of this setting the authors present technical results, including positive ones in the form of mechanisms
that are incentive compatible, individually rational and efficient, and negative ones
in the form of impossibility theorems. Monetary transfers are yet again the major tools employed in that study.

{\bf Fairness:}
As mentioned, our notion of {\em fair-share} is related to the literature  on fair division \cite{BramsTaylor}. However, it is different from maxmin fairness typically discussed in the CS literature (e.g. \cite{fairness}), and from the idea of proportional utility advocated in the literature (see \cite{BramsTaylor} for this and related concepts and discussions). In particular, our fair-share notion offers better guarantees to some players than to others. As noted earlier, our inspiration for this notion did not come from the fairness literature, but from literature on approximation algorithms~\cite{FIMN}.

{\bf Forgiveness:}
Forgiveness is an important concept in social studies. It has also been recently adopted in designing computational systems, such as reputation systems \cite{VHP}. Another line of research that introduces a notion of forgiveness is the study of evolution of cooperation \cite{Axelrod}; indeed, the famous Tit-for-Tat strategy can be viewed as employing proportional punishment and forgiveness.
While game theory does not deal explicitly with forgiveness, some classical solution concepts do incorporate a possibility of error on behalf of players. A well known example is that of {\em trembling hand} perfect equilibrium (see \cite{Selten}). It defines rational behavior of a player as a limit to which best responses converge once the probability of ``accidental" irrational behavior on behalf of other players tends to 0.

\section{Some additional proofs}

\subsection{Near deterministic jobs}

Here we prove Theorem~\ref{thm:goodneardeterministic}. We note that the proof applies even to more general classes of jobs than near-deterministic. It suffices that players know their job lengths within a power of~2.

\begin{proof} {\bf [Of Theorem~\ref{thm:goodneardeterministic}].} In our proof we prefer simplicity over presenting tight bounds. The players are asked to report their private inputs (which for near-deterministic jobs we may assume without loss of generality to be of the form $t_i$, meaning that the running time is at least $t_i$ and not more than $2t_i$).

The mechanism selects a threshold value $t$ and runs only those jobs for which $t_i \le t$. Each such job is run until completion, or aborted after $t$ steps (whichever comes first). This continues until $D$ time steps elapse. (Technically, players with $t_i > t$ may be offered a lottery of running for $D$ steps with probability $1/2n$, thus making it a dominant strategy for them to truthfully report that $t_i > t$.) We show that there is a choice of $t$ that gives the desired approximation. Order all jobs by their true lengths. The optimal schedule OPT runs a prefix of this order. Let $t/2$ be the average length of a job in this prefix. Try setting $t$ as a threshold. At least half the jobs of OPT are eligible (by Markov's inequality). However, other jobs might be eligible as well (those with $t_i \le t \le 2t_i$). We call these jobs {\em $t$-dangerous}. They form no problem if they indeed finish (each such job comes at the expense of a constant number of jobs from OPT). But the problem is that for some of them their true length is above $t$ and they are aborted. This is wasted time. If the number of dangerous jobs is not more than a multiplicative factor of $2^{\sqrt{\log n}}$ larger than the good jobs, we set the threshold at $t$. Otherwise, we double $t$. The dangerous jobs become good, but some other jobs might become dangerous. We may repeat doubling at most $\sqrt{\log n}$ times before $n$ jobs are reached. By this time, the length grew by $2^{\sqrt{\log n}}$.

Once we set the threshold at $t$, the eligible jobs are run in a random order. In expectation, at least a $2^{-\sqrt{\log n}}$ fraction of them finish, and each of them takes a time that is a factor at most $2^{\sqrt{\log n}}$ larger than the average time taken by a job in the optimal solution. This gives a factor $2^{2\sqrt{\log n}}$ approximation, as desired.

The above discussion proves the existence of a good threshold $t$. It is not hard to change this existential argument into an algorithm for computing an appropriate $t$. We omit the details.
\end{proof}

\subsection{The fair-share scheduler}
\label{sec:fairshare}

We present more details for the proof of Theorem~\ref{thm:strategicfair}. There are several variations to our fair-share scheduler, though all of them use the same basic principle of letting the player choose parameters for a lottery over time units. Here is one variation that is applicable when no job is longer than $M$, and $M$ is substantially smaller than $D$. This mechanism is $\frac{D-M}{D}$-fair, and in addition guarantees that very short jobs do finish.

Let $u$ (for unit) denote the value $u = (D-M)/n$.
Every player $P_i$ computes the value $t_i$ satisfying $u \le t_i \le M$ that maximizes his expected payoff per time unit if run for $t_i$ steps. Namely, $t_i = \max_{u \le t \le M}[f_i(t)/t]$, . The player reports the value $t_i$ to the mechanism.
The mechanism sets up an assignment problem in which jobs assign coupons to events. There are $N$ events and $n$ jobs. A job has $uN/t_i$ coupons of value $t_i/u$ each. ($N$ is assumed to be chosen such that $uN/t_i$ is integer for all $t_i$.) In a feasible assignment, a job assigns its coupon to distinct events. Let $V_j$ be the total value of coupons assigned to event $j$. We say that event $j$ is {\em overbooked} if $V_j > nu$ and moreover, this inequality continues to hold regardless of which single coupon is removed from the event. We say that event $j$ is {\em underbooked} if $V_j < nu$. Averaging shows that in every feasible assignment, if there is an overbooked event there must also be an underbooked event. Moreover, there must be at least one coupon that the overbooked event may transfer to the underbooked event while keeping the assignment feasible. As the underbooked event does not become overbooked by this transfer, repeated applications of such transfers lead to a feasible assignment with no overbooked events. In particular, $V_j \le nu + M = D$ for every $j$. Given such an assignment, the scheduler chooses an event at random and runs those jobs that assigned a coupon to this event. By the bound on $V_j$, this schedule is feasible. Moreover, each job $J_i$ is selected with probability exactly $u/t_i$, and if selected, gets $t_i$ time steps.  In particular, jobs of length $u$ (or less) are selected for sure.

{\bf Remark.} For a related definition of $u$, namely, $u = (D-M)/(n-1)$, if all $t_i$ requested by players are integer multiple of $u$, the assignment problem can be replaced by a simpler mechanism. Sort all players in order of increasing $t_i$ (breaking ties arbitrarily), and gives them consecutive intervals along the real line (starting at 0), where the length of the interval given to a player $P_i$ equals $u/t_i$. Choose random $r \in [0,1)$ uniformly at random. Select every job whose interval contains a point of the form $Z + r$, where $Z$ is integer, and run it for $t_i$ steps (or until it finishes, whatever happens first). This schedule is feasible (takes at most $D$ time steps). This follows from the observation that for every job of length $t$ that is selected, the $(t/u)-1$ players that follow it in the sorted order are necessarily not selected. Hence only the last selected job can cause surplus compared to expectation, and its length is at most $M$. Hence the total running time is at most $(n-1)u + M = D$. We also note (though this is not needed for the proof of Theorem~\ref{thm:strategicfair}) that the total length the scheduler commits to is larger than $D - 2M$. Every job of length $t$ selected is preceded by at most $t-1$ jobs not selected. Hence the loss is only in the last jobs not selected, and there are at most $(M-1)/u$ of them. Hence the time committed is at least $u(n - (M-1)/u)  > D - 2M$.

Long jobs whose optimal choice of $t_i$ is larger than $M$ are not handled by the above mechanism. To fix this, run the above mechanism with $u = D/n$ and $M = D$. The schedules obtained might not be feasible (might have length as large as $2D$), but can always be broken into two feasible schedules (one containing the largest coupon, the other containing the rest). Picking one of them at random
ensures every job at least half of its fair share. This proves Theorem~\ref{thm:strategicfair}.

An adaptive variation of the fair-share scheduler is to order players at random and then proceed in rounds (resembling the proof of item (2) of Theorem~\ref{thm:nocomplete}). The first job gets some probability (say 1/2) of running until completion (this ensures fairness with respect to long jobs), and thereafter, in each round the expectation of the lottery given to a player is the ratio between the remaining time and remaining number of players (which is $D_i/n_i$ in the notation of the proof of Theorem~\ref{thm:nocomplete}). This also leads to an $\Omega(1)$-fair schedule.

\subsection{Oblivious schedulers}

To put Theorem~\ref{thm:strategicfair} in context, it is interesting to investigate to what effect the scheduler is using the reports of the player. A scheduler that ignores the reports of players is called {\em oblivious}. We show that such a scheduler cannot be $O(1)$-fair. We break Theorem~\ref{thm:oblivious} into two statements and prove each of them separately.

\begin{proposition}
\label{thm:nooblivious}
Every oblivious scheduler is $O(1/\log n)$-fair at best, even if preemption is allowed.
\end{proposition}

\begin{proof}
Let $D = n$. There are $n$ jobs. For each job independently, its length is $2^i$ with probability $2^i/2n$, where $1 \le i \le \log n$. The expected sum of fair shares is $n\sum_{i=1}^{\log n} (2^i/2n)(1/2^i) = (\log n)/2$.

Consider an arbitrary oblivious schedule. W.l.o.g., it only schedules jobs for time periods that are powers of~2. Moreover, for every value of $t \le n$, it is more beneficial to attempt to schedule a fresh job rather than a previously scheduled job that did not finish (this is true up to a factor of~2 in the probability of terminating). Hence for simplicity we assume that the schedule only schedules fresh jobs. For such a job, the schedule need only announce a value $t$ for the number of time steps to run the job. The probability that the job indeed finishes in $t$ steps is at most $2t/2n = t/n$. The expected number of steps that are actually run is at least $t/2$ (and moreover, with probability at least $3/4$ the number of steps is at least $t/2$). Informally, this means that per time step there is probability at most $2/n$ of finishing a job, and hence the expected number of jobs that finish in $n$ time steps is at most 2. We omit further details of the proof.
\end{proof}

The bounds in Proposition~\ref{thm:nooblivious} are essentially tight.

\begin{proposition}
\label{thm:yesoblivious}
There is an oblivious scheduler (with no preemption) that is $\Omega(1/\log n)$-fair.
\end{proposition}

\begin{proof}
The scheduler selects a random threshold, much in the same way that has been done numerous times in other contexts. Pick at random an integer value $1 \le i \le \log n$ and set a threshold $T = 2^i D/n$. Sort the jobs in random order and let every job run for up to $T$ steps.

To see that this mechanism is $\Omega(1/\log n)$-fair, consider an arbitrary job $J_i$ and let $t_i$ be the value maximizing $f_i(t)/t$ (from which his fair share $(f_i(t_i)/t_i) \cdot (D/n)$ is derived). With probability $1/\log n$ we have that $T/2 < t_i \le T$, and then with probability at least $\min[1,D/nT]$ job $J_i$ gets to run for at least $t_i$ steps, and finishes with probability $f_i(t)$.
\end{proof}

Oblivious schedulers have one significant advantage over strategic schedulers - they cannot be fooled by malicious (or erroneous) reporting by the players. Hence the performance guarantees offered by oblivious schedules can form a sort of insurance against erratic reporting by players. Formally, to get this insurance, one may run an oblivious schedule (such as that of Proposition~\ref{thm:yesoblivious}) with probability 1/2, and a strategic scheduler (such as that of Theorem~\ref{thm:strategicfair}) with the remaining probability.

\subsection{Complete schedulers}

If preemption is allowed, the oblivious scheduler of Proposition~\ref{thm:yesoblivious} can easily be made complete (by repeating it in rounds until either all jobs finish or time runs out). However, it will not become $\Omega(1)$-fair. We prove Proposition~\ref{thm:complete} that provides an $\Omega(1)$-fair complete scheduler. The price to pay to achieve this combination is two-fold: preemption is used, and the solution concept is not that of dominant strategies.

\begin{proof} {\bf [Of Proposition~\ref{thm:complete}].}
We sketch the proof. The basic idea is to have two phases. In the first phase, run the fair share scheduler. This ensures $\Omega(1)$-fairness. If by the end of the first phase, there still is unused time and there still are jobs that have not finished (either because they lost their lottery or because they were aborted), do a second phase in which the remaining jobs are scheduled (in some arbitrary order) until they either all finish or $D$ time steps elapse. The problem with this idea is that because of the second phase the fair share scheduler of the first phase is no longer truthful. The number of steps that the player anticipates to get in the second phase distorts his preferred report for the first phase. For this reason, in this theorem we abandon dominant strategies and settle for Nash equilibrium. We now present the scheduler in more details.

With probability half, choose a random job and run it to completion. (This handles the fair share of long jobs.) Thereafter, perform two phases. Order the jobs at random. In the second phase, run jobs (that have not finished in the first round) to completion (or until time runs out). In the first phase, do the round version of the fair-share mechanism (see Section~\ref{sec:fairshare}) but with a twist. As in the fair-share mechanism, for each job $J_i$ allow a lottery with expectation equal to the ratio between the number of time steps remaining and $n-i$. However, now the optimal lottery for a player depends not only on its $f_i$ -- a slightly better lottery may be available if the player maximizes his overall success probability given the jobs that remain and accounting for the effect of this on his probability of getting extra time also in the second phase. Such a computation can reliably be made only if one assumes that reports of all players are truthful. If the mechanism offers to do these computations on behalf of the players, then reporting their true private inputs becomes a Nash equilibrium. In this mechanism $\Omega(1)$-fairness is inherited from the fair-share mechanism (having the second phase only improves prospects for players), and completeness is obtained by having the second phase.
\end{proof}

We sketch here an algorithmic implementation of the scheduler of Proposition~\ref{thm:complete}. The optimal choice at each step can be computed by backward induction, as follows. Recall that there are $D$ time steps and $n$ players. Given an ordering of players, fill in reverse a dynamic programming table of size $n \times D \times D \times D$. Entry $(i,j,k,L)$ says what player $P_i$ wants to play (which lottery with integer expectation) at time step $j$ if already $k$ steps are committed for the second phase (by players of index smaller than $i$ who did not finish in the first phase), and given this chosen lottery what is the probability that $L$ steps will remain after the first phase. For each value of $i$,$j$,$k$ there are $D$ possible lotteries, the best of which can be computed by backward induction. Given this best lottery, it has two possible outcomes (winning the lottery or not), and for every $L$ its probability can be computed as a weighted sum for these two outcomes (using the information previously filled in the dynamic programming table). Further details are omitted. The running time is polynomial in $D$ and $L$.

\subsection{Qualitative private inputs}

\begin{proof} {\bf [Of Theorem~\ref{thm:preemption}.]}
Each player reports a length $r_i$. For every job $J_i$, the scheduler keeps a counter $c_i$ of how many steps the job has already run, and a {\em virtual length} $v_i$ for the job. Initially, $c_i = 0$ and $v_i = r_i$. The scheduler sorts players in order of increasing order of $v_i$, always breaking ties towards the lower indexed player (or some other fixed arbitrary order that is independent of the reports $r_i$). At every time step the scheduler picks the job currently first in this order, increases its $c_i$ by~1, and runs it for one step. If the job finishes, it is removed. So far, this is similar to the shortest first scheduler (just with more notation). The new aspect happens when $c_i$ reaches $v_i$. In the shortest job scheduler, a job is aborted at this time. However, in our forgiving scheduler, at this point the job is not aborted, and instead a value of~2 is added to the current value of $v_i$. This may cause $J_i$ to move further away in the sorted order of jobs (as the order is based on the $v_i$ value). Similarly, if $J_i$ gets to perform another step, $c_i$ is increased by one, and once again $v_i$ is increased by~2, and so on. Hence the scheduler continues in this way, each time picking a job from the top of the sorted order, increasing its $c_i$ by~1, letting it run for one step, and if it does not finish and $c_i \ge r_i$ then also increase its $v_i$ by~2 and reinsert it in the sorted order. This is continued until either all jobs finish, or $D$ time steps elapse.

Let $t_i$ be the true length of a job. It will finish if and when $c_i = t_i$. Consider the value of $v_i$ at that time.  If $t_i \le r_i$ then $v_i = r_i$, Note that if $r_i > t_i$ the value of $v_i$ is exactly the true $t_i$ plus the error $r_i - t_i$ in the report. On the other hand, if $t_i > r_i$, then when $c_i$ reaches $t_i$ the value of $v_i$ is $r_i + 2(t_i - r_i) = t_i + (t_i - r_i)$. Hence it is larger than the true $t_i$ exactly by the error of $t_i - r_i$. Hence in both cases, $v_i = t_i + |r_i - t_i|$. Given the reports of all other players and the true length of their jobs, this value of $v_i$ determines where job $J_i$ will be in the sorted order at the time that it finishes, and also how many steps will every other job complete before job $J_i$ finishes. This last term in monotonically non-decreasing in the error $|r_i - t_i|$ (this is the error monotonicity property), and independent of the sign of the error (this is error symmetry). If all players happen to report true lengths for their jobs, the mechanism achieves maximum welfare (is identical to shortest first mechanism).
\end{proof}

We note that the kind of error monotonicity and error symmetry offered to the players in the proofs of Theorems~\ref{thm:preemption} and~\ref{thm:qualitative} are different. In Theorem~\ref{thm:qualitative} the guarantee is in terms of probability of success. In Theorem~\ref{thm:preemption} it is in terms of in which location $P_i$ will be in the sorted order of jobs.

We also remark that there are situations in which error monotonicity and error symmetry would not suffice in order to entice an uncertain player to report his private input. A trivial example is when $\ell_i > D$, in which case it would not make sense for the player to report anything larger than $D$. But there are other examples that are not as trivial. Hence there is plenty of room for further research (by game theorists, psychologists, sociologists) into what constitutes a mechanism that will work well in practice when players have qualitative private inputs.

\end{appendix}

\end{document}